\documentclass{llncs}
\usepackage[disable]{todonotes}
\usepackage{xspace}
\usepackage{setspace}
\usepackage{tikz}
\usetikzlibrary{arrows,shapes,trees}
\usepackage{txfonts}

\newcommand{\zigZagCurve}{zig-zag-curve\xspace}

\newcommand{\blockingPath}{blocking path\xspace}
\newcommand{\rev}{\mathit{rev}}
\newcommand{\linkageEdgeName}{linkage-edge\xspace}
\newcommand{\linkageExternalEdge}{external \linkageEdgeName}

\newcommand{\hva}{HVA}

\title{Order-preserving Drawings of Trees With Approximately Optimal Height (and
\todo{Eine KIT email waere vielleicht besser?}
Small Width)
\thanks{Work done while JB was
visiting Univ.~of Waterloo.  Research of TB supported by NSERC.} 
}
\author{
Johannes Batzill\inst{1}
\and
Therese Biedl\inst{2}
}
\institute{
Karlsruhe Institute of Technology,
Karlsruhe, Germany.  \email{batzilljohannes@gmail.com}
\and
David R.~Cheriton School of Computer
Science, University of Waterloo, Waterloo, Ontario N2L 1A2, Canada.
\email{biedl@uwaterloo.ca}
}

\begin{document}
\maketitle

\begin{abstract}
In this paper, we study how to draw trees so that they are planar, straight-line and respect a given order of edges around each node.  We focus on minimizing the height, and show that we can always achieve a height of at most $2pw(T)+1$, where $pw(T)$ (the so-called {\em pathwidth}) is a known lower bound on the height.  Hence we give an asymptotic 2-approximation algorithm.  We also create a drawing whose height is at most $3pw(T)$, but where the width can be bounded by the number of nodes.  Finally we construct trees that require height $2pw(T)+1$ in all planar order-preserving straight-line drawings.
\end{abstract}


\section{Introduction}
Let $T$ be a tree, i.e., a connected graph with $n$ nodes and $n-1$ edges.
Trees occur naturally in many applications, e.g. family trees, organizational
charts, directory structures, etc.    To be able to understand and
study such trees, it helps to create a visualization, i.e., to draw the
tree.  This is the topic of this paper.

There are many results concerning how to draw trees, see
for example \cite{BattistaF14} and the references therein.  In this
paper, we study tree-drawings of {\em ordered trees}, i.e., we assume that
with $T$ we also are given a fixed cyclic order in which the edges at each
node should occur, and our drawings should respect this.  Moreover, we
demand that the drawing is
{\em planar} (have no crossings), {\em straight-line} (edges are drawn as 
straight line segments),  and nodes are placed at points
with integer $y$-coordinates.  (We will sometimes also care about
nodes having integer $x$-coordinates.)    If all $y$-coordinates are
in the range $\{1,\dots,k\}$, then we call such a drawing a
(planar, straight-line, order-preserving)
{\em $k$-layer drawing}, and say it has {\em height} $k$  and {\em layers}
$1,\dots,k$ (from top to bottom).
We often omit ``planar, straight-line, order-preserving'', as we
study no other drawing-types. 

The main objective of this paper is to find drawings that use as few
layers as possible.  We briefly review the existing results.
For arbitrary graphs with $n$ nodes,
$\frac{2}{3}n$ layers always suffice~\cite{Chrobak199829}.
For \emph{trees}, $\log n$ layers%
\footnote{The paper bounds the width, not the height, but
rotating their drawing by 90$^\circ$ gives the result.}
are sufficient, and this is tight for some trees
\cite{Crescenzi1992187}.  Later, Suderman~\cite{Sud04}
showed that every tree can be drawn with $\lceil \frac{3}{2}pw(T) \rceil$
layers, where $pw(T)$ denotes the \emph{pathwidth} of a tree
(defined in Section \ref{sec:definitions}).  Since any tree
requires at least $pw(T)$ layers \cite{FLW03}, he hence gives
an asymptotic $\frac{3}{2}$-approximation on the number of layers required
by a tree.  Later it was shown that the minimum number of layers required 
for a tree can be found in polynomial time \cite{minimumLayer}.

All the above results were for {\em unordered} trees, i.e.,
the drawing algorithm is allowed to rearranged the subtrees
around each node arbitrarily.  In contrast to this, we study
here {\em ordered trees}, where we are given a fixed cyclic
order of edges around each node, and the drawing must be
{\em order-preserving}, i.e., respect this cyclic order.
Garg and Rusu \cite{GR03} showed that any tree has an order-reserving
drawing of height%
\addtocounter{footnote}{-1}%
\footnotemark~%
$O(\log n)$ and area $O(n\log n)$;
the height can be seen to be at most $3\log n$.

In this paper, we give a different construction for order-preserving
drawings of tree which improves the bounds of Garg and Rusu in that
we guarantee an approximation of the minimum-possible height.
Inspired by the approach of
Suderman \cite{Sud04}, we use again the pathwidth, and show
that every tree has an order-preserving drawing on $2pw(T)+1$ layers;
this is hence an asymptotic $2$-approximation algorithm on the number of layers
for order-preserving drawings.  We also show that for some trees,
we cannot hope to do better, as they need $2pw(T)+1$ layers.

In this construction, the width is potentially very large. 
We therefore give another (and in fact, much simpler)
construction that achieves $3pw(T)$ layers and for which the width
is $n$.  Since any tree has $pw(T)\leq \log_3 (2n+1)$ \cite{scheffler1990linear}, 
our results are
never worse than the ones of Garg and Rusu, and frequently better.

\section{Preliminaries}
\label{sec:definitions}
\label{section_preliminaries}

The {\em pathwidth} is a well-known graph-parameter, usually defined as
the smallest $k$ such that a super-graph of the graph is an interval graph
that can be colored with $k+1$ colors.  For trees, the following simpler
definition is equivalent \cite{Sud04}:

\begin{definition}
The {\em pathwidth} $pw(T)$ of a tree $T$ is 0 if $T$ is a single node,
and $\min_{P}\allowbreak \max_{T'\subseteq T-P}\allowbreak \left\{ 1+pw(T') \right\}$ otherwise,
where the minimum is taken over all paths $P$ in $T$.   A path where
the minimum is achieved is called a {\em main path}. 
\end{definition}

We draw trees by splitting them at a main path, drawing subtrees
recursively, and merging them.  
The following terminology is helpful.  For a tree $T$
and a strict sub-tree $C$,  a \emph{\linkageEdgeName} is an edge $e$ of $T$
with exactly one endpoint in $C$ (called the {\em linkage-node})
and the other endpoint in $T-C$ (called the {\em anchor-node}). 
Usually $C$ will be a connected component of $T-P$ for some path $P$,
and then the linkage-edge of $C$ is unique.
An {\em external linkage-edge} of a tree $T$ is an edge $e$ that
belongs to an (unspecified) super-tree $T'$ of $T$ and has exactly
one end in $T$ and the other in $T'-T$.

To be able to merge subtrees, we need to specify
conditions on subtrees, concerning not only where linkage-nodes are placed,
but also on where the external linkage-edges could be drawn such that
edge-orders are respected.

\begin{definition}
\label{def_kv_drawable}
Let $\Gamma$ be an order-preserving drawing of an ordered tree $T$, 
and let $e = (v, u)$ be an external linkage-edge of $T$ with $v\in T$.

We say that $\Gamma$ is {\em $e$-exposed} if $v$ is in the top or bottom level, and after inserting $e$ by drawing outward (up or down) from $v$, the 
drawing respects the edge-order at $v$ in the super-tree of $T$ that defined
the external linkage-edge.

We say that $\Gamma$ is {\em $e$-reachable} if 
$v$ is drawn either as unique leftmost or as unique rightmost node, and 
after inserting $e$ by drawing outward (left or right) from $v$, the
drawing respects the edge-order at $v$ in the super-tree of $T$ that defined
the external linkage-edge.
\end{definition}

See also Fig.~\ref{fig:reachable}.  We sometimes use the terms 
top-$e$-exposed, bottom-$e$-exposed, left-$e$-reachable and right-$e$-reachable
 if we want to clarify the placement
of node $v$.  Note that any top-$e$-exposed drawing can be converted into a
bottom-$e$-exposed one by rotating it by $180^\circ$;
this does not change edge orders.

\begin{figure}[t]
\hspace*{\fill}
\includegraphics[width=0.6\linewidth]{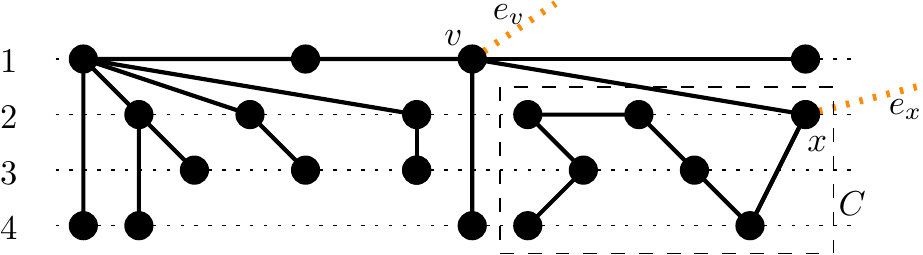}
\hspace*{\fill}
\caption{An \hva-drawing (defined in Section~\ref{sec:HVA}) on 4 layers.
$(v,x)$ is the linkage-edge of $C$, with $v$ the anchor-node and $x$ the
linkage-node.  The drawing is $e_v$-exposed (presuming the order in the
supertree is respected), but not $e_x$-exposed since $x$ is not unique among
the rightmost nodes.}
\label{fig:fig_left_reachable}
\label{fig:reachable}
\label{fig:fig_k_v_e_drawable}
\end{figure}


\section{$\mathbf{3pw(T)}$-Layer \hva-Drawings}
\label{section_3pwT}
\label{sec:HVA}

In this section, we construct special types of drawings of trees that we call
{\em \hva-drawings}:  Every edge is either Horizontal, Vertical,
or connects Adjacent layers.  We will see that for such drawings,
the width is fairly small. 
We construct such drawings using induction on the pathwidth; the
following is the hypothesis:

	\begin{lemma}
	\label{lem:3pw}
Let $T$ be an ordered tree, and let $e$ be an \linkageExternalEdge 
with endpoint $v\in T$. Then $T$ has an $e$-exposed \hva-drawing on $3pw(T)+1$ layers.
Moreover, if $T$ has at least two nodes and a main path that ends at $v$, then 
it has such a drawing on $3pw(T)$ layers.
	\end{lemma}

We first give an outline of the idea.  Exactly as in Suderman's construction
for his Lemma 7 \cite{Sud04}, we split the tree twice along paths before
recursing, choosing the paths such that they cover a main path and reach $v$.
All remaining subtrees then have 
pathwidth at most $pw(T)-1$, are hence drawn at most three units smaller
recursively, and can be merged into a drawing of these two paths.  The
main difference between our construction and Suderman's is that we must
respect the order, both within the merged subtrees and near the external
linkage-edge.  This requires a more complicated drawing for the path,
and more argumentation for why we have enough space to merge.

We phrase our main ``how to merge subtrees of a path'' as a lemma in terms
of an abstract height-bound $k$, so that we can use it twice for different
values of $k$. 
For one of these merges, it is necessary
to allow one component to be one unit taller than the others; the crux
to obtain the $3pw(T)$-bound is to realize that one such component can 
always be accommodated. 
Let $\chi(x)$ be an indicator function that is $1$ if $x$ is true and $0$ 
otherwise.

\begin{lemma}
\label{lem:mergeHVA}
Let $T$ be an ordered tree with an external linkage-edge $e_1=(v_0,v_1)$
with $v_1\in T$.  Let $P=v_1,\dots,v_l$ be a path in $T$, 
and let $C_S$ be one component of $T\setminus P$.    Fix an integer $k\geq 1$.

Assume that any component $C'$ of $T\setminus P$ has an $e'$-exposed
\hva-drawing on $k'$ layers, where $k'=k + \chi(C'{=}C_S)$ 
and $e'$ is the linkage-edge of $C'$.
Then $T$ has an $e_1$-exposed \hva-drawing on $k+2$ layers.
\end{lemma}	
\begin{proof}
We start by drawing path $P$
as a {\em battlement curve} on $k+2$ layers: Draw $(v_1,v_2)$ 
as a vertical line segment connecting the top and bottom layer,
and then alternate horizontal edges
(moving rightward) and vertical edges (to the other extreme layer).
We have a choice whether $v_1$ is in the top or bottom layer, and
do this choice such that the anchor-node $v_s$ of the special
component $C_S$ is drawn in the top layer.  Either way, $v_1$ is in
the top or bottom layer, and so $e_1$ is exposed as long as we merge
components while respecting edge-orders.

We think of the battlement curve as being extended at both ends 
with nodes $v_0$ and $v_{l+1},v_{l+2}$.  This is
done only to avoid having to describe special cases if $v_j=v_1$ or $v_j=v_l$
below;
the added edges are not included in the final drawing.  

	 \begin{figure}[t]
		\centering
		\includegraphics[page=1,width=0.9\linewidth]{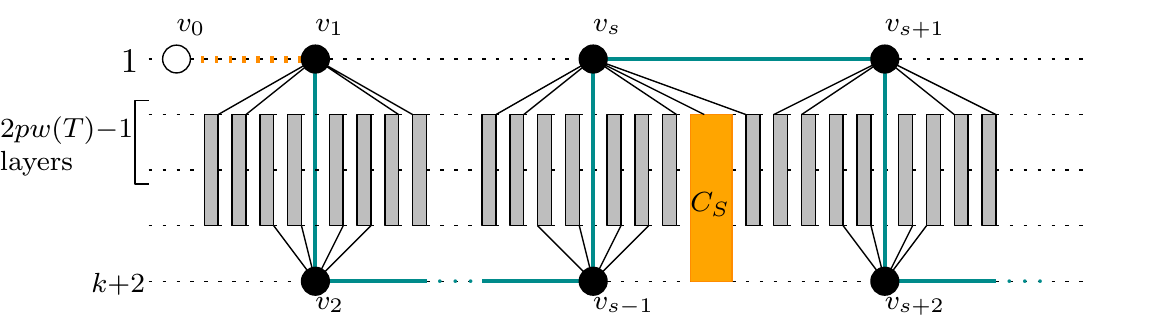}
	  	\caption{Merging at a path  (turquoise, thick) drawn as a battlement curve.}
		\label{fig:fig_3pwT_RI_plus_2}
	\end{figure}

For any component $C'$ of $T\setminus P$, the order of edges at its
anchor-node $v_j$ forces on which side
of the battlement curve $C'$ should be inserted.  More precisely,
$C'$ should be placed below the battlement curve if and only if the
linkage-edge of $C'$ appears after $(v_j,v_{j+1})$ but before $(v_j,v_{j-1})$
in the clockwise order of edges around $v_j$.  
For $j=1$, we use the edge $(v_0,v_1)=e_1$ for this choice;
since $e_1$ is drawn horizontally (because $(v_1,v_2)$ is vertical) the
edge-orders at $v_1$ are then as required for $e_1$-exposed.

Let us first assume that the drawing $\Gamma'$
of $C'$ has height at most $k$, as is the case for all components except $C_S$. 
Say $\Gamma'$ must be added below the battlement curve (adding it above
the battlement curve is symmetric).
The anchor-node $v_j$ of $C'$ is incident to a region below the battlement 
curve, say this is the region below $(v_h,v_{h+1})$ for some
$h\in \{j-2,j-1,j,j+1\}$.

Consider Fig.~\ref{fig:fig_3pwT_RI_plus_2}.  
The linkage-edge $e'$ of $C'$ is exposed in   $\Gamma'$, say it is
top-exposed and so the linkage-node of $C'$ is in the top layer.  If 
$v_j=v_h$ or $v_j=v_{h+1}$, then place $\Gamma'$ in the $k$ layers below the 
top; then $e'$ connects two adjacent layers and so  we obtain an 
HVA-drawing.   If $v_j=v_{h-1}$ or $v_j=v_{h+2}$, then
first rotate $\Gamma'$ by 180$^\circ$; this puts the linkage-node of $e'$
in the bottom layer of $\Gamma'$ and keeps all edge orders intact, and we can 
place
$\Gamma'$ in the $k$ layers above the bottom.  (We assume for this and all
later merging-steps that $\Gamma'$ has been shrunk horizontally sufficiently
so that this fits.)
If more than one component
is adjacent to $v_j$, then place these components
in the order dictated by the edge order at $v_j$.  
One easily verifies planarity,
that we have an HVA-drawing, and that the drawing is order-preserving.

It remains to explain how to deal with the special component $C_S$ whose
drawing may use $k+1$ layers.  The anchor-node $v_s$ of $C_S$ is drawn
in the top layer.  If the edge-order at $v_s$ is such that $C_S$ should be
drawn below the battlement curve, then we insert $C_S$ as before:
the bottom layer of the region below $(v_s,v_{s+1})$ is free to be
used for the drawing of $C_S$.
See Fig.~\ref{fig:fig_3pwT_RI_plus_2}.

If the edge order at $v_s$ dictates that $C_s$ should be above the 
battlement curve, then we apply 
the following {\em reversal trick}:
Let $T^{\rev}$ be the tree obtained from $T$ by reversing {\em all}
edge-orders at all nodes.  Each component can be drawn with the same
height as before, simply by flipping the drawing horizontally (which 
reverses all edges orders but keeps the linkage-edge exposed).  Apply the lemma to draw $T^{\rev}$; now the
edge order at $v_s$ is as desired.  Finally flip the drawing of $T^{\rev}$
horizontally to obtain a drawing of $T$ that satisfies all conditions.
\qed
\end{proof}

Now we are ready to prove Lemma~\ref{lem:3pw}.
We proceed by induction on $pw(T)$. In the base case, $pw(T) = 0$,
so $T$ is a single node that can be drawn on $1=3pw(T)+1$ layers;
the external linkage-edge is exposed automatically.

For the induction step, $pw(T) \geq 1$. Let $P$ be a main path of $T$,
choosing one that begins at $v$ if possible.  If $P$ does begin at $v$,
then apply Lemma~\ref{lem:mergeHVA} with this path $P$,
external linkage-edge $e_1:=e$ and $k=3pw(T)-2$.
(We have no need for a special component $C_S$ in this case.)    
Any component $C'$ of $T-P$ has pathwidth at most $pw(T)-1$,
and hence by induction can be drawn on $3(pw(T)-1)+1=3pw(T)-2= k$ layers 
with its linkage-edge exposed.  
Therefore $T$ can be drawn on $k+2=3pw(T)$ layers as desired.

Now assume that $P$ does not start at $v$, and let $R$ be the shortest path
in $T$ that starts at $v$ and ends at a node of $P$, say node $s$ is
common to $P$ and $R$.  Let $Q$ be the path containing $R$ and the part
of $P$ from $s$ to one of its ends, and let $S$ be the part of $P$ not in $R$.
See also Fig.~\ref{fig:fig_3pwT_RI_R_S_P}.

\begin{figure}[t]
\centering
\includegraphics[width=0.7\linewidth]{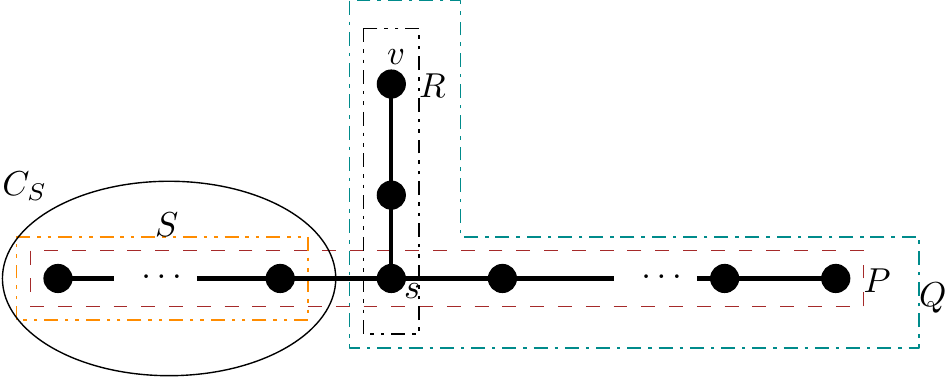}
\caption{Splitting the tree to obtain path $Q$.}
\label{fig:fig_3pwT_RI_R_S_P}
\end{figure}

Consider the components of $T-Q$.  Most of these have pathwidth at most 
$pw(T)-1$, and by induction can be drawn with height at most $3pw(T)-2$
with their linkage-edge exposed.
The one exception is the component $C_S$ that contains $S$,
which has pathwidth $pw(T)$.  Notice that $S$ is a main path of
$C_S$ that ends at the linkage-node of $C_S$, so applying
induction gives
a drawing of $C_S$ of height $3pw(C_S)=3pw(T)=k+1$ 
with its linkage-edge exposed.

Now apply Lemma~\ref{lem:mergeHVA} with path $Q$ (which ends at $v$
as required), $e_1:=e$,
$k=3pw(T)-1$, and using $C_S$ as the special component. 
This gives a drawing of $T$ of height
$k+2=3pw(T)+1$ that satisfies all properties and hence proves
Lemma~\ref{lem:3pw}.  We summarize:

\begin{theorem}
\label{thm_3pwT_RI_WxH}
Any ordered tree $T$ has an order-preserving planar straight-line \hva-drawing with height at most $\max\{1,\allowbreak 3pw(T)\}$ and width at most $|V(T)|$.
\end{theorem}
\begin{proof}
The height-bound follows immediately from Lemma~\ref{lem:3pw},
because we can (for $pw(T)\geq 1$) insert a dummy-external-linkage-edge 
at the end of a main path.
It remains to argue the width.  Observe that for any edge $(u,w)$ in the
drawing, the minimum axis-aligned rectangle $R(u,w)$ containing $u$ and $w$
is either the line segment $\overline{uw}$, or its interior is between two
layers and contains no other nodes of the drawing.  Hence an \hva-drawing
is a {\em rectangle-of-influence drawing} (see e.g.~\cite{LLMW98}).  It 
is well-known that we can change the $x$-coordinates in such a drawing without
affecting planarity, as long as relative orders are preserved.  Thus, 
enumerate all node $x$-coordinates as $x_1,\dots,x_W$, and then assign
$x(w):=i$ if node $w$ had $x$-coordinate $x_i$.  This gives another
\hva-drawing which is planar by the above, and has width at most $|V(T)|$.
\end{proof}


\section{$\mathbf{2pw(T)+1}$-Layer Drawings of Ordered Trees}
\label{section_2pwT}
	
We now improve the number of layers, at the cost of not 
having a small upper bound on the width.
Our construction is very similar to the one of Suderman for his Lemma~19
\cite{Sud04}, except that we must be more careful when merging subtrees so
that the order is preserved.  There are two key differences to the construction
from the previous section:  (1) We split three times along paths, and
achieve that the resulting subtrees have pathwidth at most $pw(T)-2$.  
(2) In the top-level split, we do {\em not} require
that the path $P$ begins the node $v$ at which the external 
linkage-edge $e$ attaches.  That makes 
the top-level split much more efficient, but means that when recursing
in the sub-tree $C_v$ that contains $v$, we now must consider {\em two} 
external linkage-edges: edge $e$ and the linkage-edge from $C_v$ to $P$.
(We make one exposed and the other reachable.)
This will complicate the induction hypothesis (which is expressed in the
following lemmas) significantly.  

\begin{lemma}
\label{lem:2pw_one}
\label{lem:2pw_two}
Let $T$ be an ordered tree and $e$ be an \linkageExternalEdge.

(a) $T$ has a drawing on $2pw(T)+1$ layers that is $e$-exposed.

(b) Let $e'$ be a second \linkageExternalEdge that has no common endpoint
with $e$.  Then $T$ has a drawing on $2pw(T)+2$ layers 
that is $e$-exposed and $e'$-reachable.
\end{lemma}

This lemma will be proved by induction on the pathwidth.  For the induction
step, we need to merge components into a drawing of a path.    Since this
will be done repeatedly with different paths, 
we phrase this merging-step as a lemma (which is similar to 
Lemma~\ref{lem:mergeHVA} but with more complicated conditions), phrasing
the height-bound as an abstract constant $k$.

	\begin{lemma}
	\label{lem:2pwMerge}
Let $T$ be an ordered tree with an \linkageExternalEdge $e_1=(v_1,v_0)$ with
$v_1\in T$.  Let $P=v_1, \dots, v_l$ be a path of $T$ starting at $v_1$.
Let $e_v=(v,u)$ be some other \linkageExternalEdge with $v\in T$.
Fix some $k\geq 1$.

Assume that every component $C'$ of $T\setminus P$ that is not $C_v$ (defined below) can be drawn on $k$ layers with its linkage-edge exposed.
Assume further that one of the following conditions holds:

\begin{enumerate}
\item $v=v_i$ for some $i > 1$, or  

\item
$v \notin P$, and the component $C_v$ of $T\setminus P$ that contains $v$ 
has a drawing on $k+1$ layers that is $e_v$-exposed and $e_C$-reachable, 
where $e_C$ is the linkage-edge of $C_v$.

\item 
$v\notin P$, and the linkage-edge $e_C$ of the above component $C_v$ is incident to $v$.
Every component $C''$ of $C_v\setminus \{v\}$ has
a drawing on $k$ layers such that the edge connecting $C''$ to $v$ is exposed.
\end{enumerate}
Then $T$ has a drawing on $k+2$ layers that is $e_v$-exposed and
$e_1$-reachable.
\end{lemma}
\begin{proof}
The first step is to draw $P$ on $k+2$ layers as a \zigZagCurve%
\footnote{Using a \zigZagCurve allows more flexibility in placing components, but means that we will not have an \hva-drawing.}
 between the top and the bottom layer, with $v_1$ leftmost.  
With this $e_1$ is the unique leftmost node and hence reachable as long as we merge components suitably. 
For ease of description, we think of the zig-zag-line as extended further left
and right with vertices $v_0$ and $v_{l+1}$; these will not be in the final
drawing.

We have the choice of placing $v_1$ in the top or in the bottom layer, and
do this as follows:  Define $v_i$ to be $v$ if $v\in P$, and 
to be the anchor-node of $C_v$ if $v\not\in P$.  Choose the placement of
$v_1$ such that $v_i$ is in the top layer.

%
The following details the {\em standard-method} of merging a
component $C'$ anchored at $v_j\in P$. 
See also Fig.~\ref{fig:fig_2pwT_plus_2}.
Assume that $v_j$ is in the top layer; the other case
is symmetric.  
Assume that the linkage-edge of $C'$ was top-exposed in the drawing
$\Gamma'$ of $C'$; else rotate $\Gamma'$ by $180^\circ$ to make it so.
Scan the edge-order around $v_j$ to find the two incident
path edges $(v_j,v_{j+1})$ and $(v_j,v_{j-1})$.  If the linkage-edge of
$C'$ appears clockwise between these two, then place 
$\Gamma'$ below edge $(v_j,v_{j+1})$, else place it above
$(v_j,v_{j+1})$.  In both cases, we do not use the top layer for $\Gamma'$,
and can hence connect to the linkage-node of $C'$ while preserving planarity
and edge-orders since the linkage-edge was top-exposed.
If multiple components are anchored at $v_j$, then
we all place them in this region, in the order as dictated by the edge-order
at $v_j$.

	 \begin{figure}[placement h]
		\centering
		\includegraphics[page=2,width=0.9\linewidth]{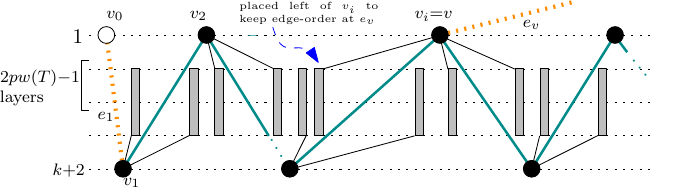}
	  	\caption{Adding components to a zig-zag path for Lemma~\ref{lem:2pwMerge}.}
		\label{fig:fig_2pwT_plus_2}
	\end{figure}

Now we show how to make $e_v$ exposed, depending on which
condition applies.

\medskip	
(1) We know that $v=v_i$ for some $i>1$ and $v_i$ is in the top layer.
After applying the reversal-trick, if needed, we may 
assume that the clockwise order
at $v{=}v_i$ in the super-tree
contains $(v,v_{i-1})$, then $e_v$, then $(v,v_{i+1})$.
Therefore, drawing $e_v$ upward from $v_i$ makes it top-exposed
as long as we merge components suitably.

Merge all components not anchored at $v_i$ with the standard-method.
For a component $C'$ anchored at $v_i=v$, the placement must be such
that the order including edge $e_v$ is also respected.  
This is done as follows (see also Fig.~\ref{fig:fig_2pwT_plus_2}):
Determine
where the linkage-edge of $C'$ falls in the clockwise order around $v$.
If it is between $e_v$ and $(v_i,v_{i+1})$, or between $(v_i,v_{i+1})$
and $(v_i,v_{i-1})$, then place $C'$ with the standard-method.  
But if it is between
$(v_i,v_{i-1})$ and $e_v$, then place the drawing of $C'$
in the region above
edge $(v_i,v_{i-1})$ (and to the right of any components anchored at
$v_{i-1}$ that may also have been placed there).  
By $i>1$, this does not place anything to the
left of $v_1$, and so $v_1$ continues to be $e_1$-reachable.

\medskip
(2) and (3):
Recall that the anchor-node $v_i$ of $C_v$ is drawn in the top layer.
Apply the reversal-trick, if needed, to ensure that $e_C$
appears between $(v_i,v_{i+1})$ and $(v_i,v_{i-1})$ in clockwise order
around $v_i$.

For (2), assume (after possible rotation)
that the drawing $\Gamma_v$ of $C_v$ is bottom-$e_v$-exposed.
Insert $\Gamma_v$ in the region below $(v_i,v_{i+1})$.
This is 
possible (after skewing $\Gamma_v$ as needed) without crossing,
since the end of $e_C$ in
$C_v$ is the unique leftmost or rightmost node of $\Gamma_v$.
		See Fig.~\ref{fig:fig_2pwT_special_construction}.

For (3), place $v$ on the bottom layer, in the area below 
edge $(v_i,v_{i+1})$, and connect it to $v_i$.  This makes $e_v$ 
bottom-exposed, as long as we are careful when placing components
of $C_v\setminus \{v\}$.  For each such component $C''$, we have a drawing
$\Gamma''$ on $k$ layers where the linkage-edge from $C''$ to $v$ is exposed.
Rotate $\Gamma''$, if needed, to make this edge bottom-exposed,
and then place $\Gamma''$ in the $k$ layers above $v$,
either left or right of edge $(v_i,v)$,
as dictated by the edge-order around $v$.  
		See Fig.~\ref{fig:fig_2pwT_special_construction}.

For both (2) and (3),
all other components $C'$ of $T\setminus P$ 
are merged with the standard-method.
This includes any other components that may be anchored at $v_i$; for
those we place them so that they are left/right of $C_v$ as dictated
by the edge-order, but still remain in the region below $(v_i,v_{i+1})$
to ensure that $v_1$ remains the unique leftmost node.
\qed
\end{proof}

	\begin{figure}[placement h]
\includegraphics[page=3,height=29mm,trim=0 0 150 0,clip]{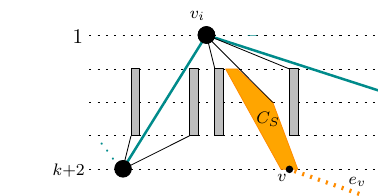}
\includegraphics[page=4,height=29mm,trim=0 0 110 0,clip]{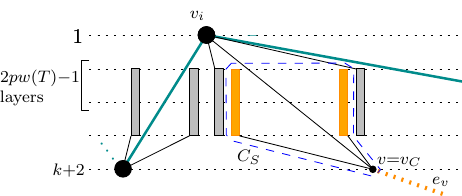}
	  	\caption{Merging component $C_v$ if condition 2 (left) or 3 (right) holds.}
		\label{fig:fig_2pwT_special_construction}
	\end{figure}

We are now ready to give the proof of 
Lemma~\ref{lem:2pw_one}.
We proceed by induction on $pw(T)$. In the base case, let $pw(T)=0$. 
Hence, $T$ is a single node and drawing $T$ on a single layer satisfies 
Claim (a).
Claim (b) is vacuously tree since any two external linkage-edges would
have the (unique) node of $T$ in common.

For the induction step let $pw(T) \geq 1$ and let $P=v_1,\dots,v_l$ be a 
main path of $T$.  Any component $C'$ of $T\setminus P$ has pathwidth 
at most $pw(T)-1$ and hence can be drawn on $2pw(T)-1$ layers with its
linkage-edge exposed by induction (Claim (a)).
For some components
we will create different drawings later to accommodate external linkage-edges.

\medskip\noindent{\bf Induction step for Claim (a):}
We distinguish cases by where the end $v$ of external linkage-edge $e$
is located.  First assume that $v\in P$.  Then 
we merge with Lemma~\ref{lem:2pwMerge} (Condition 1) using path $P$, $k=2pw(T)-1$ and $e_v:=e$.  (Use a dummy-edge at an end of $P$ as $e_1$.)
All components were drawn on at most $2pw(T)-1$ layers with their linkage-edge
exposed, so this gives a drawing on $2pw(T)+1$ layers with $e_v$ exposed.

If $v \notin P$, then let $C_v$ be the component of $T \setminus P$ that contains $v$ and let $e_C$ and $v_C$ be its linkage-edge and linkage-node. We know that 
$C_v$ has pathwidth at most $pw(T)-1$. 
If $v_C \neq v$, then
apply induction (Claim (b))
to get a drawing of $C_v$ on $2pw(T)$ layers
that is $e_v$-exposed and $e_C$-reachable. 
If $v_C = v$, then observe that any component $C''$ of $C_v\setminus \{v\}$
has pathwidth at most $pw(C_v)\leq pw(T)-1$, and by induction hence has a 
drawing
on $2pw(T)-1$ layers such that the edge from $C''$ to $v$ is exposed.  
We can hence
apply Lemma~\ref{lem:2pwMerge} (Condition 3 or 4) for path $P$,
a dummy-edge $e_1$ and $k=2pw(T)-1$ to
get the result. \qed

\medskip\noindent{\bf Induction step for Claim (b):}
Recall that $P = v_1, \dots, v_l$ is a main path of $T$ and $v'$ is the 
endpoint of edge $e'$ that should be reachable.
We now split $T$ along some paths derived from $P$ and $v'$ such that we can 
apply one of the conditions of Lemma~\ref{lem:2pwMerge}.%
\footnote{This choice of paths is the same as in Suderman,
Lemma 23, though we combine the drawings of the subtrees quite
differently to maintain edge orders.}

Fig.~\ref{fig:fig_2pwt_rel_P_L_LPrime} illustrates the following definitions.
Let $R$ be the path in $T$ from $v'$ to the nearest node of $P$; say $R$ ends at $v_s$ (possibly $v_s=v'$ and $R$ is empty). This splits $P$ into two parts $v_1, \dots, v_s$ and $v_s, \dots, v_l$. Now also consider the path $R'$ from $v'$ to $v$ (the endpoint of edge $e$ that we wish to be exposed). If $R'$ uses $v_{i+1}$ then set $Q = R \cup \{v_{i+1}, \dots, v_l\}$. If $R'$ uses $v_{i-1}$, then set $Q = R \cup \{v_1, \dots, v_{i-1}\}$. If $R'$ uses neither, then set $Q$ to any of those two. Let $S$ be the ``rest'' of $P$ not covered by $Q$, i.e., $S = \{v_1, \dots, v_{i-1}\}$ or $S = \{v_{i+1}, \dots,v_l\}$.  

The goal is to use $Q$ as the path for merging with Lemma~\ref{lem:2pwMerge}.  
However, if $S$
(the ``rest'' of $P$) is non-empty, then this is not straightforward,
because the component $C_S$ of $T\setminus Q$ that contains $S$ has pathwidth
$pw(T)$ and so is not necessarily drawn small enough.  

\begin{figure}[placement h]
\hspace*{\fill}
\includegraphics[width=0.7\linewidth]{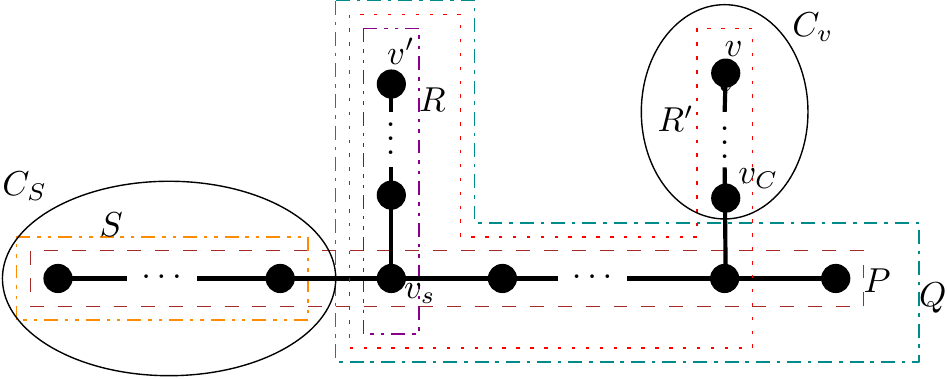}
\hspace*{\fill}
\caption{Splitting the tree to obtain path $Q$.}
		\label{fig:fig_2pwt_rel_Q_S_W}
		\label{fig:fig_2pwt_rel_P_L_LPrime}
	\end{figure}

\medskip\noindent{\bf Case 1:} $S=\emptyset$, i.e., $C_S$ is undefined.
Use Lemma~\ref{lem:2pwMerge} with $Q$ as the path, $e_1:=e'$, $e_v:=e$,
and $k=2pw(T)$.%
\footnote{For Case 1, $k=2pw(T)-1$ would have been enough, but later
cases build on top of this and then require $k=2pw(T)$.}
We must argue that this is feasible.  First,
any component $C'$ of $T\setminus Q$ has pathwidth at most
$pw(T)-1$ since $S$ is empty and so $Q$ covers the entire main path $P$.
So $C'$ has by induction (Claim (a)) a drawing on $2pw(T)-1\leq k$ layers with
its linkage-edge exposed.

If $v\in Q$ then Condition 1 holds (we know $v\neq v'$ since $e$ and $e_v$
have no end in common). 
If $v\not \in Q$ then let $C_v$ be the component of $T\setminus Q$ that 
contains $v$, and let $e_C$ and $v_C$ be its linkage-edge and linkage-node.  
We have $C_v\neq C_S$ since we chose $Q$ suitably.  
Therefore $pw(C_v)\leq pw(T)-1$.  If $v\neq v_C$, then
use induction (Claim (b)) to obtain a drawing of $C_v$
on $2pw(T)\leq k+1$ layers such that $e_v$ is exposed and $e_C$ is reachable.
So Condition 2 holds.  Finally if $v=v_C$, then any
component $C''$ of $C_v\setminus \{v\}$ has pathwidth at most 
$pw(C_v)\leq pw(T)-1$ and  by induction (Claim (a))
$C''$ can be drawn on 
$2pw(T)-1\leq k$ layers such that edge from $C''$ to $v$ is exposed.
So Condition 3 holds.  Hence regardless of the location of $v$
we obtain a drawing of
$T$ on $k+2=2pw(T)+2$ layers with $e'$ reachable and $e$ exposed.

\medskip\noindent{\bf Case 2:} $C_S$ is non-trivial, but ``belongs into a
big area'' (defined below).   Construct a drawing of $T-C_S$
as in Case 1.
We say that $C_S$ {\em belongs in the big area} if the anchor-node
of $v_S$ is in the top [bottom] layer and the clockwise [counter-clockwise]
order of edges around $v_s$ contains $(v_s,v_{s+1})$, then the linkage-edge 
of $C_S$, and then $(v_s,v_{s-1})$.  
Put differently, belonging to the big area means that the drawing of $C_S$ 
needs to be put  into a region that
has $2pw(T)+1$ levels that can be used for inserting drawings.
Construct a drawing $\Gamma_S$ of $C_S$ with its linkage-edge exposed on 
$2pw(T)+1$ layers with Claim (a).  We can insert $\Gamma_S$ with the 
standard-method for merging components since $C_S$ belongs into a big area.

\medskip\noindent{\bf Case 3:} $C_S$ is non-trivial, and does not belong
into a big area.    In this case we need a special construction to
accommodate $C_S$.%
\footnote{Because we already use the reversal-trick inside 
Lemma~\ref{lem:2pwMerge}, we cannot apply it here again.}
Let $T^-$ be the tree that results from removing
from $T$ the component $C_S$, as well as all components of $T-Q$
that are anchored at $v_s$ (the anchor-node of $C_S$).
We first construct a drawing of $T^-$ on $2pw(T)+2$ layers as in Case 1.    
Assume that $v_s$ is in the top level; the other case is symmetric.
We know that $C_S$ does not belong to a big area, so it should normally
be placed above edge $(v_s,v_{s+1})$ to preserve edge-orders.  (In the
special case that $v_s=v$, it may have to be placed above edge $(v_s,v_{s+1})$
istead to preserve edge-orders for $e_v$; this can be handled in a symmetric
fashion.)

Observe that $S$ is a main path of $C_S$.
We draw $S$ as a zig-zag-curve alternating between layer $1$ and layer 
$2pw(T)+1$, going rightwards from $v_s$.   
See Fig.~\ref{fig:fig_2pwT_insert_C_S}
Any component $C''$ of $C_S\setminus S$ has pathwidth at most
$pw(T)-1$, and can hence be drawn inductively (Claim (a)) on
$2pw(T)-1$ layers with its linkage-edge exposed.  We can hence merge
these components in the regions around $S$, exactly as in
Lemma~\ref{lem:2pwMerge}.
Finally we must merge a component $C'$ anchored at $v_s$.
If this component came (in the clockwise order around $v_s$) before
the linkage-edge of $C_S$, then 
path $S$ now blocks the connection to where we would normally
place $C'$.    (All other components at $v_s$ can be
merged with the standard-construction.)  We know that
$C'$ can be drawn with $2pw(T)-1$ layers.  Since the linkage-node of $C_S$
is placed on layer $2pw(T)+1$, we can place $C'$ in the $2pw(T)-1$ layers
below the top-row and above the linkage-edge and connect it to $v_s$ without
violating planarity and respecting edge-orders.

	 \begin{figure}[placement h]
		\centering
		\includegraphics[page=5,width=0.8\linewidth]{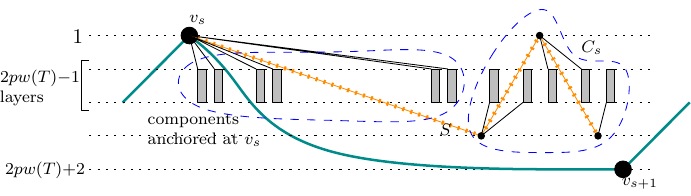}
	  	\caption{The special constructiono for component $C_S$ if it does not belong to a big area.  We draw $(v_s,v_{s+1})$ slightly curveed to avoid 
having to scale too much.}
		\label{fig:fig_2pwT_insert_C_S}
	\end{figure}

This special construction for $C_S$ does
{\em not} interfere with the (potentially special) construction for
component $C_v$ (presuming $v\not\in Q$), because
we had ensured (by using the reversal-trick, if needed) that $C_v$ belongs to
a big area. So either $C_v$ is in a different area altogether, or $C_v$
is anchored at $v_{s+1}$, and we easily keep these drawings separate.

This finishes the proof of Lemma~\ref{lem:2pw_two}.
By applying Lemma~\ref{lem:2pw_one}(a) with an arbitrary
dummy-edge as external linkage-edge, we hence obtain:
	
\begin{theorem}
\label{thm_2pwT}
\label{thm:2pw}
Any tree $T$ has a planar straight-line 
order-preserving drawing on $2pw(T) + 1$ layers.
\end{theorem}

Note that we make no claims on the width of the drawing.  In fact,
in order to fit drawings of components within the regions underneath
zig-zag-lines, we may have to scale these components horizontally
(or equivalently, widen the zig-zags significantly).  

%
We can show that the bound in Theorem~\ref{thm:2pw} is tight.  
Define an ordered tree $T_i$ recursively as follows.  $T_0$ consists
of a single node.  $T_i$ for $i>0$ consists of a path $v_1,v_2,v_3$
and 12 copies of $T_{i-1}$, three attached at each of $v_1,v_3$,
and three attached on each side of the path at $v_2$.
See also Fig.~\ref{fig:fig_2pwt_min_ti}.  By using $v_1,v_2,v_3$ as
main path one sees that $pw(T_i)\leq i$.
The following will be shown in the appendix.

\begin{figure}[placement h]
	\centering
\includegraphics[page=3,width=0.6\linewidth]{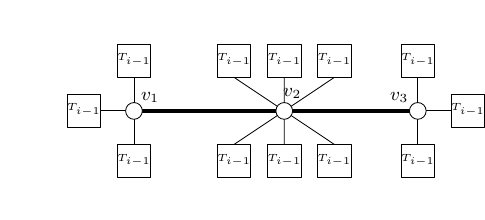}
	\caption{Tree $T_i$ has pathwidth $i$ but requires $2i+1$ layers in an order-preserving planar drawing.}
	\label{fig:fig_2pwt_min_ti}
\end{figure}

	\begin{theorem}
	\label{lem_tree_2pwt_layers}
	\label{thm:lower}
Any planar order-preserving drawing of $T_i$ has at least $2pw(T_i){+}1 = 2i{+}1$ layers.
	\end{theorem}


\section{Conclusion and Open Problems}
In this paper, we studied planar straight-line order-preserving drawings of trees that use few layers. Inspired by techniques of Suderman~\cite{Sud04}, we gave two constructions.  The first one is an asymptotic 3-approximation for the height and the width is bounded by $n$.    The second is an asymptotic 2-approximation for the height, with no bound on the width.  We also showed that `2' is tight if one uses the pathwidth for lower-bounding the height.

As for open problems, all our constructions (and all the ones by Suderman) rely on path decompositions, and hence yield only approximation algorithms to the height of tree-drawings. The algorithm for optimum-height (unordered) tree-drawings \cite{minimumLayer} uses an entirely different, direct approach. Is there a poly-time algorithm that finds optimum-height ordered tree-drawings?

\bibliographystyle{plain}
\bibliography{references,full,gd,papers}

\begin{thebibliography}{1}

\bibitem{BattistaF14}
G.~Di Battista and F.~Frati.
\newblock A survey on small-area planar graph drawing.
\newblock {\em CoRR}, abs/1410.1006, 2014.

\bibitem{Chrobak199829}
Marek Chrobak and Shin ichi Nakano.
\newblock Minimum-width grid drawings of plane graphs.
\newblock {\em Computational Geometry}, 11(1):29 -- 54, 1998.

\bibitem{Crescenzi1992187}
P.~Crescenzi, G.~Di Battista, and A.~Piperno.
\newblock A note on optimal area algorithms for upward drawings of binary
  trees.
\newblock {\em Computational Geometry}, 2(4):187 -- 200, 1992.

\bibitem{FLW03}
S.~Felsner, G.~Liotta, and S.~Wismath.
\newblock Straight-line drawings on restricted integer grids in two and three
  dimensions.
\newblock {\em J. Graph Alg. Appl}, 7(4):335--362, 2003.

\bibitem{GR03}
A.~Garg and A.~Rusu.
\newblock Area-efficient order-preserving planar straight-line drawings of
  ordered trees.
\newblock {\em Int. J. Comput. Geometry Appl.}, 13(6):487--505, 2003.

\bibitem{LLMW98}
Giuseppe Liotta, Anna Lubiw, Henk Meijer, and Sue Whitesides.
\newblock The rectangle of influence drawability problem.
\newblock {\em Comput. Geom.}, 10(1):1--22, 1998.

\bibitem{minimumLayer}
D.~Mondal, Md.~J. Alam, and Md.~S. Rahman.
\newblock Minimum-layer drawings of trees.
\newblock In {\em WALCOM: Algorithms and Computation}, volume 6552, pages
  221--232. Springer, 2011.

\bibitem{scheffler1990linear}
Petra Scheffler.
\newblock A linear algorithm for the pathwidth of trees.
\newblock In {\em Topics in combinatorics and graph theory}, pages 613--620.
  Physica-Verlag, 1990.

\bibitem{Sud04}
M.~Suderman.
\newblock Pathwidth and layered drawings of trees.
\newblock {\em Intl. J. Comp. Geom. Appl}, 14(3):203--225, 2004.

\end{thebibliography}

\newpage
\begin{appendix}

\section{$\mathbf{2pw(T)+1}$ Layers is Tight}
\label{section_2pwT_tight}

In this section, we prove Theorem~\ref{thm:lower}: the tree
$T_i$ from Fig.~\ref{fig:fig_2pwt_min_ti} requires $2i+1$ layers in
any order-preserving planar drawing.%
\footnote{The proof does not require that the drawing is straight-line;
the same lower bound holds for drawings with bends.}
We prove this by
induction on $i$; the case $i=0$ is trivial since the single-node tree $T_0$
requires 1 layer.  So assume that $i>0$ and we already know that $T_{i-1}$
requires at least $2i-1$ layers by induction.
We need a helper-lemma.

\begin{lemma}
\label{cl:Hi}
Let $H_i$ be the tree that consists of a single node $v$ with three copies
of $T_{i-1}$ attached.  Then $H_i$ requires at least $2i$ layers.
\end{lemma}
\begin{proof}
Assume to the contrary that $H_i$ could be drawn on $2i-1$ layers.  For
each copy of $T_{i-1}$, we require $2i-1$ layers.  Hence each copy of $T_{i-1}$
gives rise to a {\em blocking path} that connects the topmost and bottommost
layer and stays within that copy of $T_{i-1}$.
Add a node $v'$ above the drawing connected to the three top ends of
the three blocking paths, and a node $v''$ below the drawing connected
to the three bottom ends of the three blocking paths.  Also observe that
$v$ is connected (via a path within that copy of $T_{i-1}$)
to each of the three blocking paths.
Therefore the three blocking
paths, together with $\{v,v',v''\}$, give a planar drawing of a subdivision 
of $K_{3,3}$, an impossibility.\qed
\end{proof}

\vspace*{-5mm}
\begin{figure}[ht]
\hspace*{\fill}
\includegraphics[page=1,width=0.45\linewidth]{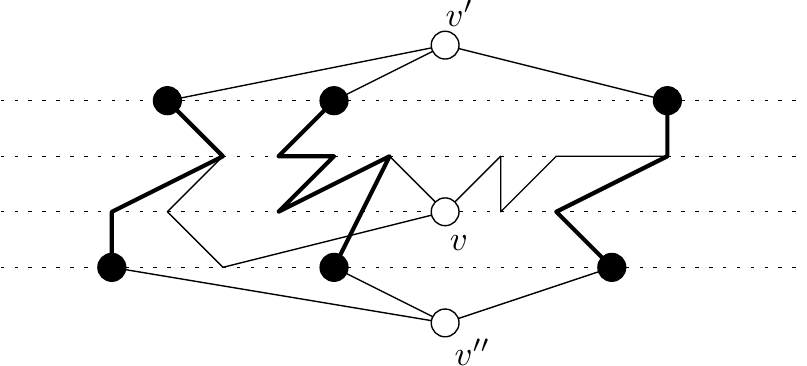}
\hspace*{\fill}
\includegraphics[page=2,width=0.45\linewidth]{lowerBounds.pdf}
\hspace*{\fill}
\caption{(Left) We can construct a planar drawing of $K_{3,3}$.  (Right) If $v_2$ is not in the top row, then the path $P$ forces a copy of $H_i$ to be drawn within $2i-1$ layers.}
\end{figure}

Now we give the induction step of the proof of 
Theorem~\ref{lem_tree_2pwt_layers}.
Since $T_i$ contains $H_i$, by Lemma~\ref{cl:Hi} it requires at least $2i$
layers.  Assume for contradiction that we have
a drawing $\Gamma$  of $T_i$ on exactly $2i$ layers.  
Let $P$ be a path that connects a leftmost node in $\Gamma$ to a 
rightmost node in 
$\Gamma$ (breaking ties arbitrarily).   It is well-known (see for example
\cite{Crescenzi1992187}) that any subtree that is node-disjoint from $P$
must be drawn either within the bottommost $2i-1$ layers or within the 
topmost $2i-1$ layer.

Observe that $P$ must contain path $v_1,v_2,v_3$, for otherwise we have
a copy of $H_i$ at one of $v_1,v_3$ that is node-disjoint from $P$ and
would be drawn in $2i-1$ layers, which is impossible.
Now consider the layer that $v_2$ is on.  Since we have $2i\geq 2$ layers,
one of the top and bottom layer does not contain $v_2$, say $v_2$ is not
on the bottom layer.  Since path $P$ uses $v_1,v_2,v_3$, and since the
drawing is order-preserving, there must be three copies of $T_{i-1}$
that are attached at $v_2$ and above path $P$, hence in the top $2i-1$
layers.  Vertex $v_2$ together with these three copies forms
an $H_i$, and since it is vertex-disjoint from $P$ (except at $v_2$,
but $v_2$ is not in the bottom layer either), it is drawn in $2i-1$ layers. 
This contradicts Lemma~\ref{cl:Hi},
so no drawing $\Gamma$ of $T_i$ on $2i$ layers can exist. \qed

\end{appendix}

\end{document}